\documentclass[journal]{IEEEtran}

\usepackage{cite}
\ifCLASSINFOpdf
\else
\fi

\usepackage{amsfonts}
\usepackage{amsthm}
\usepackage{dsfont}
\usepackage{mathrsfs}
\usepackage{amssymb}
\usepackage{color}
\usepackage{bbm}
\usepackage{epsfig}
\usepackage{amsmath}
\usepackage[english]{babel}
\newtheorem{thm}{Theorem}
\newtheorem{defn}{Definition}

\newtheorem{lem}{Lemma}

\newtheorem{prop}{Proposition}

\providecommand{\1}{\mathbf{1}}

\providecommand{\quotes}[1]{``#1"}

\providecommand{\com[2]}{\begin{tt}[#1: #2]\end{tt}}

\begin{document}
%
\title{Finite-time consensus using stochastic matrices with positive diagonals }

\author{Julien M. Hendrickx,  Guodong Shi  and  Karl H. Johansson
\thanks{
       J. M. Hendrickx is with the ICTEAM institute,  Universit\'{e} catholique de Louvain, Belgium. His work is supported by the Belgian Network DYSCO (Dynamical Systems, Control, and Optimization), funded by the Interuniversity Attraction Poles Program, initiated by the Belgian Science Policy Office. The research is also supported by the Concerted Research Action (ARC) Large Graphs and Networks of the French Community of Belgium.
 G. Shi is with College of Engineering and Computer Science, The Australian National University, ACT 0200 Canberra, Australia. K. H. Johansson is with ACCESS Linnaeus Centre, School of Electrical Engineering,
Royal Institute of Technology, Stockholm 10044, Sweden. G. Shi and K. H. Johansson are supported by the Knut
    and Alice Wallenberg Foundation, the Swedish Research Council. Email: {\tt\small  julien.hendrickx@uclouvain.be, guodong.shi@anu.edu.au, kallej@kth.se}. }}
\date{}

\maketitle

\begin{abstract}
We discuss the possibility of reaching  consensus in finite time using only linear
iterations, with the additional restrictions that the update matrices must be stochastic with positive diagonals and consistent with a given graph structure.
We show that finite-time average consensus can always be achieved for connected undirected graphs.
For directed graphs, we show some necessary conditions for finite-time consensus, including strong connectivity and the presence of a simple cycle of even length.
\end{abstract}



\section{Introduction}

\IEEEPARstart{T}{he}  problem of how a set of autonomous agents  can reach a common state via only local information exchange is widely studied. The problem becomes the  average consensus problem when the  limit is restricted to the average value of the initial states.  A standard solution is given by the  consensus algorithm \cite{JNTthesis, JNT86, jad03}, where each node iteratively updates its value as  a convex combination of the values of its neighbors.
This corresponds to a
linear dynamical system whose state-transition matrices are stochastic matrices. The asymptotic convergence of consensus algorithms has been extensively studied under various graph conditions \cite{JNTthesis, JNT86, jad03, xiao, caoming1,ren05,VBcdc,nedic09,olshevsky09,aysal09}, including some work on the optimization of the convergence rate, e.g., \cite{xiao}.
This convergence rate affects indeed the performance of several more complex algorithms using (part of) the consensus algorithms as subroutine.

Pushing this optimization to its limit leads to consensus algorithms converging in finite time. It has been shown in the literature that finite-time consensus can be reached via  continuous-time protocols \cite{cortes06,hui08,wang2010}. Quantized consensus algorithms as well converge in finite time  \cite{kashyap07, aysal08}. Discrete-time consensus algorithms converging in finite time have also been recently  discussed in  \cite{LG,kibangou,HJOG2012,guodong1a,guodong1b,ko2009matrix}, and the possibility of reaching consensus in a finite number of steps via gossiping was studied in \cite{guodong2,ko2009scheduling}.  These algorithms share several of the advantages of the centralized algorithms: They have a finite computational cost, and they guarantee that there exists a time at which all agents have exactly the same value, as opposed to approximately the same value. Actually, it has  been shown that distributed algorithms converging in finite time in some settings are faster than any possible centralized algorithm \cite{LG,kibangou,HJOG2012}.

In this paper, we investigate  finite-time convergence for consensus algorithms defined by a product of stochastic matrices with positive diagonal entries. The positivity condition means that agents always give positive weights to their own states when computing their new states. This natural condition is widely imposed in the existing literature on consensus algorithms, e.g.,   \cite{JNT86,jad03,ren05,caoming1,VBcdc,nedic09,olshevsky09,guodong1a,guodong1b}, and is for example automatically satisfied by any algorithm representing the sampling of a continuous-time process. In the absence of  positivity condition, deciding whether consensus is reached  becomes a fundamentally hard problem \cite{haj, VBAO}. The restriction to stochastic matrices with positive diagonal entries is one of the main differences between our work and the results in \cite{LG,kibangou,HJOG2012}, as they allow general real matrices, as long as they are consistent with the graph under consideration.
Similarly, the authors of \cite{ko2009matrix} only require their matrices to be column-stochastic (i.e. nonnegative and having each column summing to 1) in order to preserve the average value of $x$, but not necessarily stochastic. Some of the algorithms that they obtain do however also satisfy our requirements, as will be explained in Section \ref{sec:undirected_graphs}.

The problem we consider is also related to the finite-time consensus computation problems \cite{chris,yuanye}, where computing the consensus limit in finite steps from a given asymptotically convergent algorithm  was considered.  Compared to the problem considered in this paper, those methods require more memory and node computations.

We now introduce  the  problem under consideration. A matrix $A\in \Re^{n\times n}$ is \emph{stochastic} if it is nonnegative and  $A\1 = \1$, i.e., the elements of any of its row sum to one. We say that a stochastic matrix is \emph{consistent with a graph $\mathcal{G}(V,E)$} with ${V}=\{1,\dots,n\}$ if $A_{ij}>0$ for $i\neq j$ only if $(j,i)\in E$.
We insist on the fact that the presence of the edge $(j,i)$ does not require $A_{ij}$ to be positive, but only allows it. We say that $A$ has a \emph{positive diagonal} if $A_{ii}>0$ for every $i$. Finally, we use $v'$ to denote the transpose of a vector $v$ in order to avoid ambiguities with the finite time $T$.
The first problem that we consider is finite-time consensus.

\begin{defn}\label{def:finite_time_consensus}
The sequence of stochastic matrices $(A_1,A_2,\dots,A_T)$ with positive diagonal achieves \emph{finite-time consensus} on a graph $\mathcal{G}$, if $A_t$  is consistent with $\mathcal{G}$ for $t=1,\dots,T$ and
$A_TA_{T-1}\dots A_2A_1 \in \mathbf {S}$, where $\mathbf{S}$ denotes the set of rank-one stochastic  matrices in $\Re^{n\times n}$,  i.e., matrices of the form $\1v'$, for some nonnegative vector $v$ whose entries sum to one. \end{defn}
So, if a sequence of stochastic matrices $(A_1,A_2,\dots,A_T)$ with positive diagonal achieves finite-time consensus, the iteration $x(t) = A_t x(t-1)$ reaches $x(T)= x^*\1$ for every $x(0)=(x_1(0),\dots,x_n(0))'\in \Re^n$, for some $x^*\in \Re$ that depends on $x(0)$. If  $x^*$ is always the average value of $x_1(0),\dots,x_n(0)$, i.e., equal to $\frac{1}{n}\1'x(0)$, then we say that the matrix sequence achieves finite-time average consensus.

\begin{defn}
The sequence of matrices $(A_1,A_2,\dots,A_T)$ with positive diagonal achieves \emph{finite-time average consensus} on a graph $\mathcal{G}$,  if $A_t$  is consistent with $\mathcal{G}$ for $t=1,\dots,T$ and
$A_TA_{T-1}\dots A_2A_1 = \frac{1}{n} \1\1'$.
\end{defn}

The outline for the rest of the paper is as follows. In Section 2 we  show that finite-time average consensus can always be achieved on connected bidirectional graphs. Then  Section 3  discusses  directed graphs, for which finite-time consensus  is far more challenging. We present three necessary conditions for finite-time consensus on directed graphs, and an example of a directed graph for which finite-time consensus can be achieved. Finally some concluding remarks are given in Section 4.

\section{Undirected graphs}\label{sec:undirected_graphs}

We show that finite-time average consensus can always be achieved on undirected graphs.
This result could actually also be obtained by an application of an algorithm of Ko and Gao \cite{ko2009matrix}, developped independently of this work and with a different approach. In \cite{ko2009matrix}, average consensus in finite time is reached by having first one node obtaining the average of all nodes' values, while preserving the global average constant. Then, this node is excluded from further interactions, and the procedure is successively repeated on all other nodes in an appropriate order. Our proof, on the other hand, relies on recursively building a set of agents at (average) consensus, and growing this \quotes{island of consensus} by successively adding all the nodes.

\begin{thm}\label{thm:undirected_trees}
If $\mathcal{G}$ contains a bidirectional spanning tree, then there exists a sequence of at most $n(n-1)/2$ stochastic matrices with positive diagonal that achieves average consensus on $\mathcal{G}$. In particular, finite-time average consensus can be achieved on every undirected graph.
\end{thm}
\begin{proof}
We show by recurrence that finite-time average consensus can be reached on a bidirectional tree $\mathcal{G}_T$ in $n(n-1)/2$ steps, which will complete  the proof since every edge that does not belong to the  bidirectional spanning tree of $\mathcal{G}$ can just be ignored.

The result trivially holds if the tree $\mathcal{G}_T$ contains only one node. Let us suppose now that it contains $n+1 \geq 2$ nodes, and select a leaf node (i.e., a node with degree one) which we call $v_0$. By our recurrence assumption, average consensus can be reached for nodes $\mathcal{V}\setminus\{v_0\}$ in $T \leq n(n-1)/2$ steps since the graph obtained by removing node $v_0$ from $\mathcal{G}_T$  is  a connected tree of size $n$. Let us suppose that suitable matrices have been chosen  so that for every $i\in \mathcal{V}\setminus\{v_0\}$ there holds $x_i(T) = \bar x_{V\setminus\{v_0\}}(0)=\frac{1}{n}\sum_{j\in V\setminus\{v_0\}} x_j(0)$.

We assume now that $x_{v_0}(T) = x_{v_0}(0) = 1$ and  $x_{V\setminus\{v_0\}}(0) = -\frac{1}{n}$, so that their average is $1 + n\times\frac{-1}{n}=0$. We are going to find a sequence of $n$ (or less) stochastic matrices with positive diagonal consistent with $\mathcal{G}_T$ that drive all states to zero in finite time.

We denote by ${\rm diam}(\mathcal{G}_T)$ the diameter of $\mathcal{G}_T$, i.e., the largest distance between any two nodes of $\mathcal{G}_T$. In particular, every node is at a distance at most ${\rm diam}(\mathcal{G}_T)$ from $v_0$. For $k = 0,1,\dots, {\rm diam}(\mathcal{G}_T)-1$, let $V_k$ be the set of nodes at distance $k$ from $v_0$ on the tree $\mathcal{G}_T$ that are not leaves, i.e., $V_0= \{v_0\}$, $V_1$ is the set non-leaf neighbors of $v_0$, $V_2$ the set of non-leaf neighbors of nodes in $V_1$ that do not belong to $V_1$ nor $V_0$, etc. Let  $L_k$ be the set of leaves at distance $k$ from $v_0$ for $k=1,\dots, {\rm diam}(\mathcal{G}_T)$. See Figure \ref{fig:ex_tree} for an illustration.

\begin{figure}
\centering
\includegraphics[scale= .6]{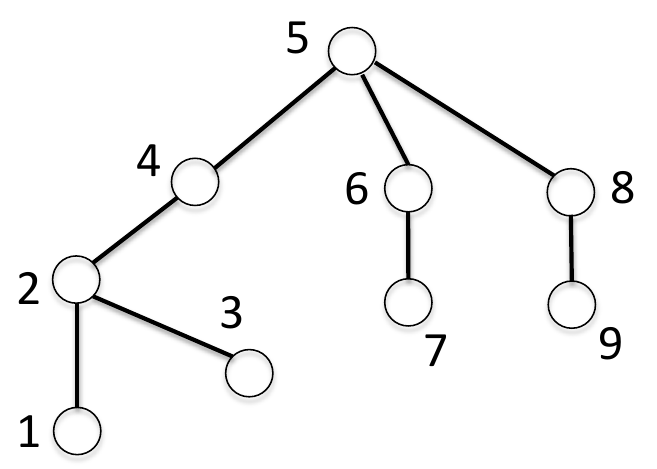}
\caption{Illustration for the proof of Theorem  \ref{thm:undirected_trees}. Suppose $v_0=1$ is selected. Then for the tree given in the figure we have $V_1=\{2\}$, $V_2=\{4\}$, $V_3=\{5\}$, $V_4=\{6,8\}$, $L_1=\{\emptyset\}$,  $L_2=\{3\}$,  $L_3=L_4=\{\emptyset\}$, and $L_5=\{7,9\}$. Clearly every node in $V_k$ and $L_k$, if any, is connected to only one node of $V_{k-1}$.  }\label{fig:ex_tree}
\end{figure}

Observe that for every $k>0$, every node in $V_k$ and $L_k$, if non-empty,  is connected to exactly one node of $V_{k-1}$, as follows from the following argument: The existence of at least one neighbor in $V_{k-1}$ follows from the definition of distance. On the other hand, no node of $V_k$ or $L_k$ can be connected to two nodes of $V_{k-1}$, as that would form a cycle in $\mathcal{G}_T$, which is impossible since  $\mathcal{G}_T$ is a tree. Nodes in $V_k$ are also connected to at least one node in $V_{k+1}$ or $L_{k+1}$. Indeed, they are by definition not leaves, so they must be connected to at least one other node than the one in $V_{k-1}$, but they cannot be connected to any other node in $V_k$ or $V_{k-1}$, for that would create a cycle. Besides, connections between nodes whose distance to $v_0$ differ by more than one are by definition impossible. Finally, for every $k$ all nodes in $V_k$ and $L_k$ share a common value at time $T$ ($1$ for $V_0$ and $-1/n$ for the others).

We now show that the following evolution of  $x_i(t), t\geq T$ can be achieved by multiplying $x$ by ${\rm diam}(\mathcal{G}_T)$ stochastic matrices with positive diagonal consistent with $\mathcal{G}_T$:
\begin{itemize}
\item If $i\in V_k$, then (i) $x_i(t) = -1/n$ for $T \leq t < T +k$; (ii) $x_i(T+k)  =1 /2^{k}$; (iii); $x_i(t) = 0$ for $t>T+k$.
\item If $i$ is a leaf in $L_k$, then (i) $x_i(t) = -1/n$ for $T \leq t < T +k$; (ii) $x_i(t) = 0$ for $t\geq T+k$.
\end{itemize}
We show this by recurrence on $k$. For $k=0$, the situation corresponds to that of our initial recurrence assumption. Suppose now that it holds for $k-1$ and let us consider step $k$ ($k>0$). Only nodes in $L_k$, $V_k$ and $V_{k-1}$ change their values, so other nodes need not be considered.

Nodes in $L_k$ and $V_k$ have value $-1/n$ at time $T+k-1$. As argued above, each of them is connected to a node in $V_{k-1}$, who has a value $1/2^{k-1}$ at time $T+k-1$ by the recurrence hypothesis. Therefore, the new value zero of the nodes in $L_k$ at time $T+k$
lies strictly between their former value $-1/n$ and the former value $1/2^{k-1}$ of their neighbors in $V_{k-1}$. The value is thus equal to a weighted average of these two values with positive coefficients. The same argument applies also for the nodes in  $V_k$. In other words, the desired $x(T+k)$ can be reached by multiplying $x(T+k-1)$ with a stochastic matrix consistent with $\mathcal{G}_T$ with positive diagonal.

Consider now the nodes of $V_{k-1}$. Their value at $T+k-1$ is $1/2^{k-1}$, and their new value at $t+k$ is zero. As argued above, each of these nodes has at least one neighbor in $V_k$ or $L_k$, whose value at time $T+k-1$ is thus $-1/n$. The new value zero at time $T+k$ of nodes in $V_{k-1}$ lies thus strictly between their former value $1/2^{k-1}$ and that of the neighbors in $V_k$ or $L_k$. It can therefore be reached by multiplication by a stochastic matrix consistent with $\mathcal{G}_T$ with positive diagonals, which completes the proof of the recurrence.

We have thus shown the existence of  $A_{T+1}, A_{T+2},\dots,A_{T+{\rm diam}(\mathcal{G}_T)}$ with positive diagonal and consistent with $\mathcal{G}_T$ such that $A_{T+{\rm diam}(\mathcal{G}_T)}\cdots A_{T+2}A_{T+1} x(T) = 0$ if $x_{v_0}(T) = 1$ and $x_i(T) = -1/n$ for every $i\in  V\setminus\{v_0\}$. Using linearity and the fact that $A\1=\1$ for stochastic
matrices,   it follows that under the recurrence assumption $x_i(T) = \bar x_{V\setminus\{v_0\}}(0)=\frac{1}{n}\sum_{j\in V\setminus\{v_0\}} x_j(0)$ for all $i\in V\setminus\{v_0\}$,  that $A_{T+{\rm diam}(\mathcal{G}_T)}\dots A_{T+2}A_{T+1} x(T)  = \bar x_{V}(0)\1$. Average consensus is thus achieved on $\mathcal{G}_T$ in $T + {\rm diam}(\mathcal{G}_T)$ steps. Using the recurrence
assumption $T\leq n(n-1)/2$ and the bound ${\rm diam}(\mathcal{G}_T)\leq n$ for a graph of $n+1$ nodes, it follows that average consensus is achieved in at most $\frac{1}{2}n(n-1)+ n = \frac{1}{2}n(n+1)$ steps on any tree of $n+1$ nodes, which completes our proof.
\end{proof}
By a small modification of the proof, one can actually show that it is possible to reach any weighted average of the initial conditions with positive weights in the same number of steps.

\section{Directed graphs}\label{sec:directed_graphs}

Theorem \ref{thm:undirected_trees} shows that finite-time (average) consensus can always be achieved on a directed graph if it contains a bidirectional spanning tree. We will now see that the situation is much more complex when the graph is \quotes{essentially} directed and does not contain a bidirectional spanning tree. We begin by providing certain necessary conditions.

As is well known in the literature \cite{xiao,ren05}, the existence of a  directed spanning tree for a directed graph $\mathcal{G}$  is a necessary and sufficient condition for finding an asymptotically convergent  consensus algorithm on  $\mathcal{G}$. In our next result, we show that that strong connectivity is necessary  for finite-time consensus.

\begin{prop}
There exists a sequence of stochastic matrices with positive diagonal that ensures  finite-time consensus on a graph $\mathcal{G}$  only if $\mathcal{G}$ is  strongly connected.
\end{prop}
\begin{proof}

Suppose $\mathcal{G}$ is not strongly connected. Then there exist two subsets $V_1,V_2$ of nodes such that no edge leaving $V_2$ arrives in $V_1$.
Let us take as initial condition $x_i(0) =0$ for all $i\in V_1$ and $x_i(0) = 1$ for all $i\in V_2$, and consider an arbitrary sequence $(A_1,\dots, A_T)$ of stochastic matrices with positive diagonal consistent with $\mathcal{G}$. Since there is no edge from $V_2$ to $V_1$, $[A_t]_{ij} = 0$ for any  $i\in V_1,j\in V_2$  so that the values of the nodes $V_1$ are never influenced by those of the nodes in $V_2$.
Therefore, we have $x_i(t)=0,\  t\geq1,\ i\in V_1.$

We introduce ${h}(t)=\min \{x_i(t):\ i\in V_2\}$. Denote $a_t^\ast =\min \Big\{ \big[A_t \big]_{ii}:\  i\in V_2 \Big\}.$
Then it is easy to see that
\begin{align}
x_i(t+1)&=\sum_{j=1}^n [A_t]_{ij}x_j(t)\nonumber\\&\geq [A_t]_{ii}x_i(t)+\Big(1-[A_t]_{ii} \Big)\min \{x_m(t):\ m\in V\}\nonumber\\
&\geq [A_t]_{ii}{h}(t)\nonumber\\
&\geq a_t^\ast {h}(t)\nonumber
\end{align}
for all $i\in V_2$ and $t$. Thus, we have
${h}(t+1)\geq a_t^\ast {h}(t)$ for all $t\geq 0$,
which implies
\begin{align}
{h}(T)\geq {h}(0) \prod_{t=1}^T a_t^\ast = \prod_{t=1}^T a_t^\ast >0=x_m(T),\ m\in V_1.\nonumber
\end{align}
Therefore, consensus cannot be achieved  by any finite sequence of stochastic matrices with positive diagonals consistent with $\mathcal{G}$ for the initial condition that we have considered.
\end{proof}


We now show that achieving finite-time consensus  requires the presence of a cycle of even length. This does not contradict the tree-based result of Theorem \ref{thm:undirected_trees}, as every pair of opposite edges of a bidirectional graph constitute a directed cycle of length 2.

\begin{prop}\label{prop:cycle_even_length_needed}
There exists a sequence of stochastic matrices with positive diagonal that ensures  finite-time consensus on a graph $\mathcal{G}$   only if $\mathcal{G}$ contains a simple directed cycle with even length.
\end{prop}
\begin{proof}
Suppose that finite-time consensus can be reached on  graph $\mathcal{G}$ in $T$ steps. Consider particular initial conditions $x(0)$, and let $x^*$ be the  consensus value, i.e., $x_j(T)=x^*$ for all $j$. Let  $i_0$ be a node reaching the final value only at the last step, i.e., $x_{i_0}(T-1) \neq x^*$. We suppose without loss of generality that $x_{i_0}(T-1) > x^*$. By definition, $x_{i_0}(T)$ is a convex combination of the values $x_j(T-1)$ of the neighbors $j$ of ${i_0}$ and of $x_{i_0}(T-1)$, with a positive weight for the latter value. Since $x_{i_0}(T-1)>x^*$ and $x_{i_0}(T) = x^*$, there must exist a  neighbor $i_1$ of $i_0$ for which $x_{i_1}(T-1) < x^*$.

By a similar argument, there exists a neighbor $i_2$ of $i_1$ such that $x_{i_2}(T-1) > x^*$.
Doing this iteratively, we can build an arbitrary long sequence of indexes $i_k$ such that $x_{i_k}(T-1) > x^*$ if $k$ is even, and $x_{i_k}(T-1) < x^*$ if $k$ is odd, and where the node
$i_{k+1}$ is a neighbor of $i_k$, as shown in Figure \ref{fig:ex_even cycle}. Since there are only finitely many nodes in the graph, some indices are repeated in this sequence. Let $j^*$ be the first node who is repeated twice in the sequence. By construction of the sequence and of $j^*$, there is a path from $j^*$ to itself passing no more than once any other node. Moreover, this path must be of even length. Otherwise, one of the two first indices $k$ corresponding to $j^*$ would be even, and the other odd, so that we would have simultaneously $x_{j^*}(T-1) >x^*$ and $x_{j^*}(T-1) <x^*$, which is impossible. This completes the proof.
\end{proof}

\begin{figure}
\centering
\includegraphics[scale= .7]{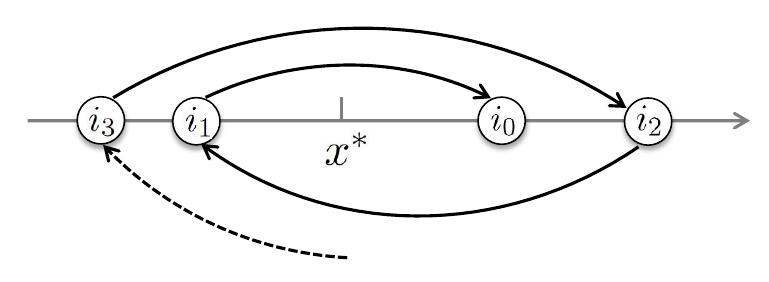}
\caption{Illustration of the proof of Proposition \ref{prop:cycle_even_length_needed}. The node values are sorted at time $T-1$. A simple directed cycle with even length can be constructed if consensus is reached at time $T$.}\label{fig:ex_even cycle}
\end{figure}

Note that all nodes in the cycle of even length in Proposition \ref{prop:cycle_even_length_needed} reach the final value at the last time step.

The presence of a cycle of even length is a necessary condition for finite-time consensus, but is certainly not sufficient. Actually, the next result states that finite-time consensus cannot be achieved if the graph only consists of a cycle of even length.

\begin{prop}\label{prop:cycles_impossibility}
Suppose $\mathcal{G}$ is a simple directed cycle. Then  no finite sequence of stochastic matrices with positive diagonals achieves consensus   on $\mathcal{G}$.
\end{prop}

Without loss of generality, we will restrict our attention to  a cycle $C_n$ of $n$ nodes, where there is a directed edge $(i,i-1)$ for $i=2,\dots,n$ and an edge $(1,n)$. Moreover, we identify $x_{n+1}$ with $x_1$: if $i=n$, then $x_{i+1}$ denotes $x_1$. Similarly, if $i=1$, $x_{i-1}$ denotes $x_n$. To prove Proposition~\ref{prop:cycles_impossibility}, we need the following intermediate result, showing that the presence of two consecutive nodes with the same sign is preserved by multiplication by a stochastic matrix consistent with $C_n$  and with positive diagonal when $n$ is even.

\begin{lem}\label{lem:consecutive_signs}
Let $C_n$ be a cycle of even length $n$ and $A$ a stochastic matrix with positive diagonal consistent with $C_n$. Let $x\in \Re^n$ and $y=Ax$.
If there is  $i\in V$ such that $x_i,x_{i+1}\geq 0$ or $x_i,x_{i+1}\leq 0 $, then there is  $j\in V$ such that $y_j,y_{j+1}\geq 0$ or $y_j,y_{j+1}\leq 0 $.
\end{lem}
\begin{proof}
Observe first that since $A$ is a stochastic matrix consistent with $C_n$ with positive diagonals, it holds that $y_i = \alpha_i x_i + (1-\alpha_i) x_{i+1}$ for some $\alpha_i\in (0,1]$ for every $i$. As a consequence, the following implications and there symmetric versions for opposite signs hold:

 (a) If $x_i,x_{i+1}\geq 0$, then $y_i\geq 0$.

 (b) If $x_i \geq 0$ and $y_i <0$, then $x_{i+1}<0$.\\

Let us now assume without loss of generality that $x_1,x_2\geq 0$. It follows from  implication (a) that $y_1\geq 0$.  If $y_2\geq 0$ then the result holds  with $j=1$. Otherwise, $y_2 <0$, and it follows from implication (b) above that $x_3 <0$. Now if $y_3 \leq 0$, then the result holds with $j=2$ since $y_2<0$. Otherwise, $y_3 >0$, which by (b) implies that $x_{4} >0$. By repeating this argument and using the fact that $n$ is even, we see that either the result holds for some $j$, or $x_n >0$ and $y_{n+1} = y_1 <0$, in contradiction with our initial assumptions.
\end{proof}

We  now prove Proposition \ref{prop:cycles_impossibility}.
\begin{proof}
If the number of nodes $n$ is odd, the result follows directly from Proposition \ref{prop:cycle_even_length_needed}. Let us thus assume that $n$ is even, and suppose that there exists a sequence of $T$ stochastic matrices with positive diagonals consistent with $C_n$ guaranteeing  finite-time consensus. We consider the following initial condition: $x_1(0) = x_2(0) =1$, and $x_i(0) = 0$ for every other $i$; and we denote by $x^*$ the consensus value that the system reaches for this initial condition. Clearly $x^*\leq 1$, so that $x_1(0) \geq x^*$ and $x_2(0)\geq x^*$. By applying
Lemma \ref{lem:consecutive_signs} recursively to $x(t) - x^*\1$, we see that for any time $t\leq T$, and in particular for $t=T-1$, there exists $j$ such that either $x_j(t)-x^*\geq 0$ and $x_{j+1}(t)-x^*\geq 0$ or $x_j(t)-x^*\leq 0 $ and $x_{j+1}(t)-x^*\leq 0$.

On the other hand,  the proof of Proposition \ref{prop:cycle_even_length_needed} shows that if consensus is reached at iteration $T$ on a value $x^*$, the graph must contain a cycle whose nodes have values at time $T-1$ that are all different from $x^*$, and for which the sign of $x_i(T-1)-x^*$ are opposite for any two consecutive nodes on the cycle. Since the only cycle of $C_n$ is the whole graph itself, this means that for every $i$, $x_i(T-1)-x^*$ and $x_{i+1}(T-1) -x^*$ are nonzero and have opposite signs.

We thus obtain a  contradiction, which implies that consensus in finite time cannot be achieved for cycles of even length.
\end{proof}


We have thus proved so far that finite-time consensus can be achieved on a directed graph $\mathcal{G}$ only if it is strongly connected and contains a simple cycle of even length, and that it cannot be achieved if it only consists of a cycle of even length.
The combination of these  impossibility results might suggest that finite-time consensus can never be achieved unless the graph contains a bidirectional spanning tree or is \quotes{equivalent} in some sense to such a graph. This is however not true. Consider the example in Figure \ref{fig:ex_directed_graph}, consisting of a directed cycle of length 4 to which is added one bidirectional edge between nodes 1 and 3. One can verify that the following matrices are consistent with the graph
$$
A_1=A_3=\frac{1}{2}\begin{pmatrix}
1&  1 & 0 & 0 \\
0&  1 & 1 & 0\\
0 & 0 & 1 & 1\\
1 & 0 & 0 & 1
\end{pmatrix} ,\hspace{.1cm}
A_2=A_4= \frac{1}{2}\begin{pmatrix}
1 & 0 & 1 & 0 \\
0 & 1 & 1 & 0\\
1 & 0 & 1 & 0\\
1 & 0 & 0 &1
\end{pmatrix},
$$
and that $A_4A_3A_2A_1 = \frac{1}{4}\1\1'$, so that finite-time average consensus can be achieved.

\begin{figure}
\centering
\includegraphics[scale= .6]{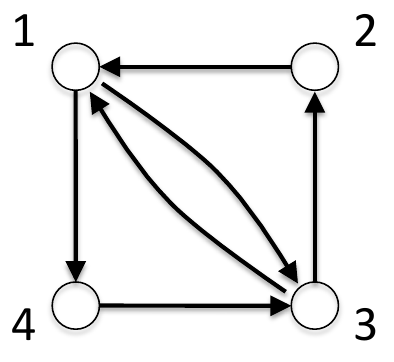}
\caption{Example of directed graph on which finite-time average consensus can be achieved, despite the fact that it does not have a bidirectional spanning tree. }\label{fig:ex_directed_graph}
\end{figure}

\section{Conclusions and open questions}\label{sec:ccl}
This paper discussed  the existence of finite-time convergent (average) consensus algorithms.

We have provided a new proof that (average) consensus can always be achieved by a finite sequence of matrices on every connected undirected graph.
For directed graphs, we have proven that finite-time consensus is reachable  only if
the graph is strongly connected and contains a simple directed cycle with even length, but that it cannot be reached if the graph only consists of such a directed cycle.
This shows that requiring all diagonal elements to be positive  reduces the set of graphs on which finite-time consensus or average consensus can be reached. An adaptation of the \quotes{gather and distribute} method described in Section 4.2 of \cite{LG} shows indeed that without this requirement, finite-time average consensus can be reached for any strongly connected graph.

Note that our impossibility proofs never use the fact that the sequence of matrices must drive the system to consensus for every initial condition. So our impossibility results also hold in the more general case where the matrix $A_t$ can be chosen as a function of $x(t-1)$.

Finally, we have also provided an example of a directed graph where finite-time average consensus can be achieved.
The necessary condition combined with the example suggest that the precise conditions under which finite-time consensus can be achieved over a general directed graph could be intricate.


\begin{thebibliography}{1}


\bibitem{JNTthesis}J. N. Tsitsiklis, ``Problems in decentralized decision making and computation," Ph.D. thesis, Dept. of Electrical Engineering and
Computer Science, Massachusetts Institute of Technology, 1984.

\bibitem{JNT86}
J. N. Tsitsiklis, D. Bertsekas, and M. Athans, ``Distributed asynchronous
deterministic and stochastic gradient optimization algorithms," {\em
IEEE Trans. Automatic Control}, vol. 31, no. 9, pp. 803-812, 1986.


\bibitem{jad03}
A. Jadbabaie, J. Lin, and A. S. Morse,
\newblock ``Coordination of groups of mobile autonomous agents using nearest neighbor rules,"
\newblock {\em IEEE Trans. Automatic Control}, vol. 48, no. 6, pp. 988-1001, 2003.


\bibitem{xiao} L. Xiao and S. Boyd, ``Fast linear iterations for distributed averaging," {\em Systems and Control Letters}, vol. 53, pp. 65-78,  2004.


\bibitem{ren05} W. Ren and R. Beard, ``Consensus seeking in multiagent systems under dynamically
changing interaction topologies," {\em IEEE Trans. on Automatic Control}, vol. 50, no. 5, pp. 655-661,
2005.

\bibitem{VBcdc} V. Blondel, J. M. Hendrickx, A. Olshevsky and J. Tsitsiklis, ``Convergence in multiagent coordination, consensus, and flocking," {\em IEEE Conf. Decision and Control}, pp. 2996-3000, 2005.





\bibitem{caoming1} M. Cao,  A. S. Morse and B. D. O. Anderson, ``Reaching a consensus in a dynamically changing
environment: a graphical approach," \newblock {\em SIAM J. Control Optim.}, vol. 47, no. 2, pp. 575-600, 2008.



\bibitem{nedic09} A. Nedi\'{c}, A. Olshevsky, A. Ozdaglar, and J. N. Tsitsiklis, ``On distributed
averaging algorithms and quantization effects," {\it IEEE Trans.
Automatic Control}, vol. 54, no. 11, pp. 2506-2517, 2009.

\bibitem{olshevsky09} A. Olshevsky and J. N. Tsitsiklis, ``Convergence speed in distributed consensus and averaging,"  {\em  SIAM J. Control Optim.}, vol. 48, no.1, pp. 33-55, 2009.

\bibitem{aysal09} T. C. Aysal, M. E. Yildiz, A. D. Sarwate, and A. Scaglione, ``Broadcast gossip algorithms for consensus," {\em IEEE Trans. Signal Processing}, vol. 57, no.7, pp. 2748-2761, 2009.





\bibitem{cortes06} J. Cort\'{e}s, ``Finite-time convergent gradient flows with applications to
network consensus," {\em  Automatica}, vol. 42, no. 11, pp. 1993-2000, 2006.

\bibitem{hui08} Q. Hui, W. M. Haddad, and S. P. Bhat, ``Finite-time semistability and consensus for nonlinear dynamical networks," {\em  IEEE Trans. Autom. Contr.}, vol.53, pp. 1887-1900, 2008.


\bibitem{wang2010}  L. Wang and F. Xiao, ``Finite-time consensus problems for
networks of dynamic agents," {\em IEEE Trans. Autom. Contr}, vol. 55, no.4, pp. 950-955, 2010.

\bibitem{kashyap07} A. Kashyap, T. Ba\c{s}ar, and R. Srikant, ``Quantized consensus," {\em Automatica}, vol. 43, pp. 1192-1203, 2007.

\bibitem{aysal08} T. C. Aysal, M. J. Coates,  and M. G. Rabbat, ``Distributed average consensus
with dithered quantization," {\em IEEE Trans. Signal Processing}, vol. 56, no.10, pp. 4905-4981, 2008.


\bibitem{LG} L. Georgopoulos, ``Definitive consensus for distributed data inference," PhD thesis, EPFL, Lausanne, 2011.

\bibitem{kibangou} A. Y. Kibangou, ``Finite-time consensus based protocol for distributed estimation over AWGN channels," in {\em  Proceedings of the Joint 50th IEEE Conference on Decision and Control and European Control Conference}, pp. 5595-5600, December 2011.

\bibitem{HJOG2012} J. M. Hendrickx, R. M. Jungers, A. Olshevsky and G. Vankeerberghen, ``Graph diameter, eigenvalues, and minimum-time consensus",  {\em Automatica}, vol. 50, no.2, pp 635-640, 2014.


\bibitem{guodong1a} G. Shi and K. H. Johansson. ``Convergence of distributed averaging and maximizing algorithms: Part I: time-dependent graphs,” in {\em American Control Conference}, pp. 6096-6101, Washington, DC, 2013.

\bibitem{guodong1b} G. Shi and K. H. Johansson. ``Convergence of distributed averaging and maximizing algorithms: Part II: state-dependent graphs.” in {\em American Control Conference}, pp. 6859-6864, Washington, DC, 2013.
 


\bibitem{ko2009matrix} C.-K. Ko and X. Gao, ``On matrix factorization and finite-time average-consensus,"  in {\em Proceedings of the 48th IEEE Conference on Decision and Control 2009 held jointly with the 2009 28th Chinese Control Conference. CDC/CCC 2009}, pp. 5798-5803, 2009.


\bibitem{guodong2} G. Shi, M. Johansson and K. H. Johansson, ``When do gossip algorithms converge in finite time?" in {\em Proceedings of the 21st International Symposium on Mathematical Theory of Networks and Systems (MTNS'14)},  pp. 474-478, 2014.



\bibitem{ko2009scheduling} C.-K. Ko and L. Shi, ``Scheduling for finite time consensus," in {\em Proceedings of the American Control Conference, 2009. ACC'09}, pp. 1982--1986, 2009.




\bibitem{haj} J. Hajnal, ``Weak ergodicity in non-homogeneous Markov chains," {\em Proc. Cambridge
Philos. Soc.}, no. 54, pp. 233-246, 1958.



\bibitem{VBAO} V.  Blondel and A. Olshevsky, ``On the cost of deciding consensus,"  to appear in {\em SIAM Journal on Control and Optimization}, {\tt arXiv: 1203.3167}, 2012.







\bibitem{chris} S. Sundaram and C. N. Hadjicostis, ``Finite-time distributed consensus
in graphs with time-invariant topologies," in {\em Proceedings of American
Control Conference}, pp.711-716, 2007.

\bibitem{yuanye} Y. Yuan, G. Stan, L. Shi, M. Barahona and J. Goncalves, ``Decentralised minimum-time consensus," {\em Automatica}, in press.


\bibitem{julien2011} J. M. Hendrickx, A. Olshevsky, and J. N. Tsitsiklis, ``Distributed anonymous discrete function computation,"  {\it IEEE Trans.
Autom. Control}, vol. 56, no. 10, pp. 2276-2289, 2011.











\end{thebibliography}
\end{document}